\documentclass{article}

\usepackage{graphicx} 
\usepackage{amsmath,amssymb}
\usepackage{booktabs}
\usepackage{amsthm}
\usepackage{xcolor}
\usepackage{verbatim}
\usepackage{tcolorbox}
\usepackage{enumitem}
\usepackage{lipsum}
\usepackage{comment}
\usepackage{float} 
\usepackage{tikz}
\usetikzlibrary{positioning, arrows.meta}
\usepackage{authblk}
\usepackage{natbib}
\usepackage{orcidlink} 

\title{\textbf{Anchor-proofness in Voting}}

\author[1]{Federico Fioravanti\orcidlink{0000-0002-7154-1192}\thanks{Corresponding author: federico.fioravanti@univ-st-etienne.fr}}
\author[2]{Zoi Terzopoulou\orcidlink{0000-0001-5289-434X}}
\affil[1]{Université Jean Monnet Saint-Étienne, CNRS, Université Lyon 2, emlyon business school, GATE, 42023 Saint-Étienne, France.}
\affil[2]{CNRS, Université Jean Monnet Saint-Étienne, Université Lyon 2, emlyon business school, GATE, 42023 Saint-Étienne, France.}
\date{}

\theoremstyle{definition}  
\newtheorem{teo}{Theorem}
\newtheorem{propi}{Proposition}
\newtheorem{lemma}{Lemma}
\newtheorem{coro}{Corollary}
\newtheorem{myrem}{Remark}

\newtheorem{defn}{Definition}
\newtheorem{myex}{Example} 
\newenvironment{exa}{\begin{myex}\rm}{\hfill$\vartriangle$\end{myex}} 
\newenvironment{remark}{\begin{myrem}\rm}{\hfill$\vartriangle$\end{myrem}} 
\newcommand{\pref}{p}
\newcommand{\prof}[1]{\boldsymbol{#1}}
\newcommand{\alt}{X}
\newcommand{\app}{A(\alt)}
\newcommand{\info}{\mathit{info}}

\begin{document}

\maketitle

\begin{abstract}
We study how order-dependent individual evaluations affect collective decision-making. We consider a voting environment in which voters have intrinsic preference rankings and acceptability thresholds, but form approval ballots sequentially as alternatives are presented by a social planner. Because alternatives encountered early on serve as anchors, approval ballots may depend on the presentation order even when intrinsic preferences remain unchanged. We ask when approval-based voting rules are robust to this order dependence, a property we call anchor-proofness.

We first show that no non-constant approval-based voting rule is anchor-proof. We then study weaker forms of robustness and characterize anchor-proof profiles for the classical approval voting rule and the nomination rule. Finally, we examine the strategic role of the social planner under different information structures. We show that, while a planner with no information about voters’ preferences cannot safely manipulate any rule satisfying weak and total unanimity, even limited information about voters’ acceptability or plurality assessments can restore manipulation. Overall, we highlight presentation order as an important design feature of collective decision-procedures when voters are subject to anchoring.

\end{abstract}


\section{Introduction}
A large literature in economics and psychology documents systematic cognitive biases in individual decision making \citep{tversky1974judgment}, including the \emph{anchoring bias}, whereby early information serves as a reference point that affects subsequent evaluations. The \emph{contrast effect} can be viewed as a particular manifestation of this bias \citep{Tversky_Griffin_contrasteffect2000}. The consequences of such biases for welfare, learning, and market behavior are well documented \citep[see, for example,][]{THALER198039,kahnemanthalerfairness1986,bordalosalience2012}. Comparatively less is known, however, about how individual-level biases affect collective decision-making. In particular, when individual evaluations are aggregated to determine a collective outcome, a bias that affects the formation of such individual evaluations may also affect the outcome of the group decision. Understanding this link requires studying how the procedure through which individual evaluations are formed interacts with the aggregation rule. To illustrate, consider Example~\ref{ex:pb}.

\begin{exa}
\label{ex:pb}
In participatory budgeting platforms, citizens consider different projects and approve certain ones to receive funding in their neighborhood.\footnote{See for example the Paris platform \texttt{www.paris.fr/budget-participatif}.} The order in which projects are presented can then become a relevant feature of the collective decision-process. Imagine a simple scenario where a voter is presented with only two projects $a$ and $b$, she considers both acceptable, but she intrinsically prefers project $a$ to project~$b$. If she encounters project~$a$ first, this may set a high \emph{anchor} against which project~$b$ is subsequently evaluated, and so the voter may approve $a$ but reject $b$. If instead she encounters project~$b$ first, she may approve both projects given that she finds them acceptable. Thus, the same intrinsic preferences can give rise to different approval ballots solely because of the presentation order in combination with an anchoring bias.

This individual-level order dependence can translate into different collective outcomes. Consider a group of five voters, two of whom prefer project $a$ to~$b$ and three of whom prefer project $b$ to~$a$. If everyone encounters their top alternative first, and hence approves only that alternative, then $b$ receives the most approval votes, as shown in Table~\ref{tab:before-anchor}.

\begin{table}[!ht]
\centering
\begin{minipage}{0.45\textwidth}
\centering
\begin{tabular}{ccccc}
1  & 2 & 3 & 4 & 5 \\
\toprule
$a$  & $a$ & $b$  & $b$ & $b$  \\
\bottomrule
\multicolumn{5}{c}{approval winner:  $b$}
\end{tabular}
\caption{Everyone sees their top alternative first.}
\label{tab:before-anchor}
\end{minipage}
\hfill
\begin{minipage}{0.45\textwidth}
\centering
\begin{tabular}{ccccc}
1  & 2 & 3 & 4 & 5 \\
\toprule
$a$ & $a$ & $a,b$ & $a,b$ & $a,b$  \\
\bottomrule
\multicolumn{5}{c}{approval winner:  $a$}
\end{tabular}
\caption{Everyone sees first alternative~$a$ and then alternative~$b$.}
\label{tab:after-anchor}
\end{minipage}
\end{table}

 If instead all voters encounter $a$ first and then $b$, the first two voters approve only $a$, while the remaining three approve both alternatives. Consequently, $a$ receives the most approval votes, as shown in Table~\ref{tab:after-anchor}. So, changing only the presentation order can reverse the collective outcome while leaving voters' intrinsic preferences unchanged.
\end{exa}

Sequential presentation of alternatives is common both in institutional contexts of collective decisions and in digital-choice architectures, where ordering is an implicit feature of the decision process.\footnote{In modern online platforms that rely on recommender systems, the ordering presented to users is often personalized. On the other hand, electoral contexts usually present alternatives in a fixed order for all voters. Some digital competitions of participatory budgeting incorporate randomized orders.} In such environments, anchoring can affect the formation of individual approval ballots by determining the reference point against which newly encountered alternatives are evaluated. Consequently, the ballots submitted to the aggregation rule may reflect not only voters' intrinsic preferences and acceptability evaluations, but also the procedure through which alternatives are presented. This raises a fundamental question for collective decision-making: can a voting rule produce outcomes that are robust to changes in the presentation order?

Randomizing the presentation order of the alternatives could be considered one way to mitigate the influence of anchoring in expectation. However, randomization does not necessarily provide a satisfactory solution in collective-choice settings. In electoral contexts, for instance, randomized procedures may raise concerns about transparency, uniformity, and  legitimacy. More fundamentally, a collective decision-procedure may be required to be robust in every realized instance rather than only on average across randomizations. Motivated by this perspective, we study a worst-case notion of robustness and ask whether collective outcomes remain invariant to deterministic presentation orders.

\paragraph{Our Contribution.}
We work within a framework of approval voting, where voters have intrinsic ordinal preferences over a finite set of alternatives and also hold an acceptability threshold over the alternatives. A key feature of our model is that approval ballots are generated from preferences via a simple \emph{anchoring} mechanism: a social planner specifies the order in which a voter encounters the alternatives, and the voter approves an alternative if and only if it is acceptable and strictly preferred to all previously encountered alternatives. Hence, the presentation order determines the mapping from intrinsic preferences to approval ballots.

Our analysis addresses two central questions. First, we study whether approval-based voting rules can be robust to anchoring. We call a rule \emph{anchor-proof} if the set of selected winners is invariant with respect to the presentation order. Such rules would guarantee that the collective outcome is unaffected by the way in which a social planner presents the alternatives. We show that every non-constant rule violates anchor-proofness in its most demanding form. We then identify conditions under which anchor-proofness can nevertheless arise for restricted preference structures or particular presentation orders.

Second, we study the strategic role of the social planner when the voting rule is not anchor-proof. A planner who has preferences over collective outcomes may exploit her control over the presentation order to influence the decision. We therefore examine how the possibility of such manipulation depends on the information available to the planner about voters' intrinsic preferences. Our results identify cases in which manipulation is possible and characterize how it varies with the voting rule and the planner's information. In particular, although limited information about voters' preferences can suffice for manipulation, we show that complete lack of information prevents a planner from safely manipulating voting rules satisfying weak and total unanimity.

\paragraph{Related Literature.}

A well-developed strand of the social choice literature studies the relationship between voters' underlying ordinal preferences and the approval ballots they cast. This relationship has been made explicit, for example, by \citet{endriss2009preference}. Following the seminal work on approval voting by \citet{brams1983approval}, much of this literature assumes that a sincere voter with an intrinsic ordinal preference forms a \emph{consistent} approval ballot, in which every approved alternative is preferred to every disapproved alternative \citep{barrot2013possible,endriss2013sincerity}. This consistency assumption is central to results comparing approval-based rules with ranking-based rules \citep{gehrlein1981borda,regenwetter1998choosing,regenwetter2004approval}. \citet{Terzopoulou2025} study the converse problem, exploring how ranking ballots can be generated from approval preferences, again under the assumption that approved alternatives are always preferred to disapproved ones. More recently, \citet{garcialaprestamartinezpanero2025} and \citet{alcantud2026} consider extensions of preference-approval structures in group decision making, including voting systems that combine approval information with ranking information. Our framework starts from a different premise: rather than taking preference-approval structures as primitive, we derive approval ballots from intrinsic preferences through a sequential anchoring process.

A closely related line of work investigates how limitations on voters' ability to access their full preferences affect the formation of ballots. The order in which alternatives are encountered determines the subset over which ballots are formed and the input to the aggregation rule consists of partial orders \citep{Terzopoulou2023}. The distinction with our framework is important. Instead of reporting ballots that are faithful restrictions of their intrinsic preferences, in our setting voters eventually encounter all alternatives, but the presentation order determines \emph{how} these alternatives are evaluated. Anchoring therefore remains relevant even without capacity constraints.

Our behavioral mechanism is specifically motivated by the literature on contrast effects \citep{Tversky_Griffin_contrasteffect2000}. Contrast effects have been documented across a wide range of evaluative settings, including judgments of crime severity and punishment \citep{pepitone1976contrast,RODRIGUEZconstrast2016107}, sentencing and jury decisions \citep{leibovitch2016relative,bindlerpathdependency2019}, financial markets \citep{hartzmarkfinancialmarkets2018}, speed dating \citep{bhargavafisman2014}, housing demand \citep{simonsohnhousingdemand2006}, commuting decisions \citep{simonsohncommuters2006}, judgments of facial attractiveness \citep{kramer2013sequential}, synchronized diving \citep{kramerdiving2017}, and performance evaluation on the Idol series \citep{PAGEidolseries2010186}. The evidence closest to our model is provided by \citet{radbruchschiprowski2025}, who exploit quasi-random assignment of candidates to evaluators and time slots and find that an individual's assessment decreases with the quality of the other candidates seen by the same evaluator. These findings provide behavioral support for the idea that evaluations can depend systematically on the presentation order during a decision process.

Our mechanism is also reminiscent of models of order-dependent choice. A literature initiated by \citet{rubinsteinsalant2006} studies choice from lists, where the order in which alternatives are presented can affect individual choice. 
Our framework differs from these models in that voters do not select a single alternative from a sequence; rather, they sequentially evaluate alternatives and construct an approval ballot. A related issue is order dependence arising from decisions previously taken rather than from alternatives previously encountered; for example, \citet{chendecmakundergambler2016} attribute such effects to the \emph{gambler's fallacy}. 

Order effects are notably captured in large-scale choice environments. Eye-tracking experiments show that users disproportionately trust and select results appearing early in Google searches, even when those results are objectively less relevant \citep{Pan2007}. Similarly, randomized variation in the ranking of products on Expedia shows that higher-ranked products attract greater consumer attention and clicks \citep{Ursu2018}. More directly regarding collective decision making, a large literature discusses ballot-order effects in elections \citep{millerkrosnick1998,koppellsteen2004,hoimai2008,meredithsalant2013}. 

Our work is further relevant for questions on \emph{agenda manipulation}, which ask how the order of collective comparisons can affect the outcome of group decision-making processes. Originating with \citet{farqu1969theoryofvoting} and \citet{black1958theoryofcommitteeselections}, this literature characterizes the alternatives that an agenda-setter can induce a committee to select under complete information \citep{miller1977graphtheoretical,banks1985sophisticated}, with subsequent extensions to richer settings \citep{barberagerber2017seqvotagenmanip,roessler2018collectivecommitment}. More recent work considers agenda manipulation and ballot design under incomplete information \citep{kleinermoldovanu2017contentbasedseqvoting,gershkovmoldovanushi2019votingwhattoputontheballot}. Our framework differs from this literature in that the presentation order may vary across voters, voters encounter individual alternatives rather than pairwise comparisons, and aggregation takes place only after all approval ballots have been formed. Thus, the planner's influence operates through the interaction between the design of the presentation procedure and order-dependent individual evaluations.

Finally, there is growing interest in the area of \emph{behavioral social choice}, which examines how cognitive constraints shape individual preferences and consequently collective outcomes. \citet{regenwetter2009behavioural} argue that standard models of social choice presuppose stable and context-independent evaluations, whereas empirical evidence consistently exhibits context sensitivity and framing effects. They call for theoretical frameworks that incorporate behavioral regularities. \citet{bonnefon2010behavioral} shows that even preferences over aggregation procedures are malleable. Using controlled experimental settings of judgment aggregation, he demonstrates that a simple positive or negative framing of the judgment at hand can reverse individuals' procedural preferences, which can in turn shift the collective outcome. \cite{Baujard2011} discuss how voting behavior depends on the particular electoral system in the context of French presidential elections. \citet{binder2015campaigns} provide complementary field evidence using a natural experiment in California ballot initiatives. They show that framing effects meaningfully alter electoral support for low-salience propositions, while vigorous campaign activity substantially attenuates these effects.

\paragraph{Roadmap.} 
Section~\ref{sec:model} introduces our proposed voting model that incorporates an anchoring mechanism.
Section~\ref{sec:anchorproof} presents an exhaustive analysis of the different variations of our central property---anchor-proofness---and provides respective impossibility and possibility results for approval-based voting rules.
Section~\ref{sec:strategy} explores the manipulation by a social planner under potentially limited information and obtains relevant characterization results.
Section~\ref{sec:conc} concludes.


\section{Model}\label{sec:model}

Let $N=\{1,\ldots,n\}$ be a finite set of \emph{voters} and $\alt$ a finite set of \emph{alternatives}, with $|\alt|=m >2$. For an arbitrary set $Y$, we denote by $L(Y)$ the set of all linear orders, i.e., complete and acyclic binary relations, over $Y$.  We call a linear order $\pref \in L(\alt)$ a \emph{ranking} over the set of alternatives $\alt$.
If $(x,y)\in \pref$,  we will say that alternative $x$ is preferred to $y$ in $\pref$. For simplicity, we often write $p=(p^1, p^2, \ldots, p^m)$; that is, $p^k$ is the alternative that appears in the $k^{\mathit{th}}$ position of the ranking.  The most preferred alternative in a ranking appears in position~1, and the least preferred alternative holds position~$m$.
Inspired by the preference-approval model introduced by \citet{BramsSanver2009}, each voter $i\in N$ is endowed with an intrinsic \emph{preference} over the set of alternatives:\footnote{For simplicity, we use the symbol $\pref$ to denote both a linear order and a preference; the intended meaning will be clear from the context.} 
$$\pref_i = (p_i^1,\ldots, p_i^m)$$  
A preference $\pref_i \in P(\alt)$ is a ranking over the alternatives associated with an \emph{acceptability threshold} $t_i \in \{1,\ldots, m\} $ that captures the maximum position that voter~$i$ considers acceptable for approving the corresponding alternative in the ranking.\footnote{We assume that at least one alternative is acceptable for each voter, which ensures that voters submit non-empty approval ballots. This is not crucial for our results.} Formally, $P(\alt)= L(\alt)\times \{1,\ldots, m\}$.
Whenever the acceptability threshold is not relevant, we will treat $\pref_i$  simply as a linear order. So, an alternative~$x$ with position $k$ in the ranking $\pref_i$ is \emph{acceptable} if $k \leq t_i$.  
Let $\mathit{ACC}(\pref_i) = \{\pref_i^1,\ldots,\pref_i^{t_i}\} \subseteq \alt$ be the set of acceptable alternatives by voter~$i$ with preference~$\pref_i$.  
Note that the most preferred alternative of a voter is always acceptable by that voter.

We denote by $\boldsymbol{\pref} = (\pref_1,\ldots, \pref_n)\in P(\alt)^N$  the profile of  preferences of all voters in $N$. 
When $\mathit{ACC}(\pref_i)= \alt$ (equivalently $t_i=m$), we say that the preference is \emph{tolerant}.
Moreover, if for all $i\in N$ it is the case that $\mathit{ACC}(\pref_i)= \alt$, we call the profile tolerant. 
We define $\mathit{tol}(\pref_i) \in P(\alt)$ to be the \emph{tolerant} version of $\pref_i$: the associated ranking of $\mathit{tol}(\pref_i)$ is the same as of $\pref_i$ and moreover all alternatives are acceptable in $\mathit{tol}(\pref_i)$.  
The tolerant version of a profile $\prof{\pref}=(\pref_1,\ldots, \pref_n)$ is $\prof{\mathit{tol}(\pref)} = (\mathit{tol}(\pref_1),\ldots, \mathit{tol}(\pref_n))\in P(\alt)^N$.
When $|\mathit{ACC}(\pref_i)|=1$ (equivalently $t_i=1$), we say the preference is \emph{intolerant}.
If for all $i\in N$ it is the case that $|\mathit{ACC}(\pref_i)|=1$, we call the profile intolerant.

Given a profile $\prof{\pref}$ and an alternative $x \in \alt$, we define the \emph{plurality} and the \emph{acceptability} points of $x$ respectively, as follows: 
$$ plur_{\prof{\pref}}(x)=|\{i\in N\mid x=\pref_i^1\}|$$
$$ acc_{\prof{\pref}}(x)=|\{i\in N\mid x\in \mathit{ACC}(p_i)\}|$$

Thus, the plurality points of $x$ equals to the number of voters that rank $x$ first in $\prof{\pref}$; while the acceptability points of $x$ equals the number of voters that accept $x$ in $\prof{\pref}$. We also use the following notation:
\[ \mathit{PLUR}(\prof{\pref}) =\{x\in \alt \mid x=\pref_i^1 \text{ for some } i\in N \}\]
\[ \mathit{ACC}(\prof{\pref}) =\{x\in \alt \mid x\in \mathit{ACC}(\pref_i) \text{ for some } i\in N \}\]
In words, $\mathit{PLUR}(\prof{\pref})$ captures the set of alternatives that are ranked first by at least one voter in $\prof{\pref}$, and $\mathit{ACC}(\prof{\pref})$  the set of alternatives that are acceptable by at least one voter in $\prof{\pref}$.

\subsection{Anchoring Mechanism}

Let $S(\alt)$ be the set of all permutations over the set of alternatives.\footnote{Technically, $S(\alt)$ is the same mathematical entity as $L(\alt)$, though we introduce separate notation to highlight the two conceptually different notions.} Each voter~$i\in N$ is shown the alternatives in steps, one at a time, according to the order prescribed by her corresponding permutation $\sigma_i$ in the order vector $\boldsymbol{\sigma} = (\sigma_1,\ldots, \sigma_n)\in S(\alt)^N$.
Given a voter~$i$ and an order~$\sigma_i$, let $\sigma_i^k\in \alt$ be the alternative that is presented to her in step~$k$. 
Let $\app$ be the set of all non-empty approval ballots over the set of alternatives---that is, $\app = 2^{\alt}\setminus \{\emptyset\}$.  
Voter $i$ with preference~$\pref_i$ creates an approval ballot $A_{\pref_i,\sigma_i} \in \app$ based on the order~$\sigma_i$ via the following iterative procedure, where $k > 1$:
\begin{center}
\begin{tcolorbox}[colframe=black, colback=white, boxrule=0.8pt, width=1.02\textwidth]
\begin{itemize}
    \item[ Step 1:] $A_{\pref_i,\sigma_i}^1=\{\sigma_i^1\} \cap \mathit{ACC}(\pref_i)$. 
    \item[ Step $k$:]
     $A_{\pref_i,\sigma_i}^k=\begin{cases}
         ( \{\sigma_i^k\} \cap ACC(p_i) ) \cup A^{k-1}_{\pref_i,\sigma_i} & \text{ if } (\sigma_i^k,x) \in \pref_i \text{ for all } x\in A_{\pref_i,\sigma_i}^{k-1}\\
         A_{\pref_i,\sigma_i}^{k-1} & \text{ otherwise }
     \end{cases}$
     \item[ Finally:] $A_{\pref_i,\sigma_i}=A_{\pref_i,\sigma_i}^m$
    \end{itemize}
\end{tcolorbox}
\end{center}

This means that a voter constructs her ballot by examining the alternatives sequentially, one at a time, and with commitment: once an alternative is approved, it cannot be revoked later. When the voter encounters the $k^{\text{th}}$ alternative in the presentation order, for every $k$, she exhibits an anchoring bias: she approves the alternative if and only if she deems it acceptable and she has not previously encountered another more preferred alternative. Said differently, at every step~$k$ the anchor is the best alternative encountered so far. It is important to note that voters here are assumed to be sincere---i.e., it is only their bias and not any form of incentive that affects their ballot formation.

Observe also that it is not the case that every approval ballot $A_{\pref_i,\sigma_i}$ formed by a preference~$\pref_i$  and an order~$\sigma_i$ is consistent, although this is common in prior literature \citep[for example,][]{BramsSanver2009}. Under the anchoring process, an alternative outside a voter's approval ballot may be more preferred by that voter to some alternative inside the ballot (see Example~\ref{ex: non-adm}).

\begin{exa} \label{ex: non-adm}
Consider a tolerant preference $\pref_i = (x,y,z)$ and an order $\sigma_i=(z,x,y)$. 
Then, $A_{\pref_i,\sigma_i}^1=\{z\}$ (because $z \in \mathit{ACC}(\pref_i)$),  $A_{\pref_i,\sigma_i}^2=\{x,z\}$ (because $(x,z) \in \pref_i$), and $A_{\pref_i,\sigma_i}^3= A_{\pref_i,\sigma_i}^2= \{ x,z\}$ (because $x \in A_{\pref_i,\sigma_i}^2$ but $(y,x) \notin \pref_i$). 
Thus, the voter approves $z$ but disapproves $y$ although she intrinsically prefers $y$ to $z$, generating the approval ballot $A_{\pref_i,\sigma_i}=\{ x,z\}$. 
This happens because $z$ is acceptable and is examined at the very beginning, so it is approved since no better anchor exists at the time; but the voter has already seen the more preferred alternative $x$ when $y$ is examined, so $y$ is disapproved since it is compared to a better anchor.
\end{exa}

We next introduce two particular orders that depend on a preference~$\pref_i$ of voter~$i$:
\begin{itemize}
    \item 
$\underline{\sigma_i}$ is the order that presents to voter~$i$ the alternatives from least to most preferred in $\pref_i$. Formally, $\underline{\sigma_i} = (p_i^m,p_i^{m-1},\ldots, p_i^1)$. Thus:
$$ A{\pref_i, \underline{\sigma_i}} = \mathit{ACC}(\pref_i)$$
    \item
$\overline{\sigma_i}$ is the order that presents to voter~$i$ the alternatives from most to least preferred in $\pref_i$. Formally, $\overline{\sigma_i} = (p_i^1,p_i^2,\ldots, p_i^m)$. Thus:
$$A{\pref_i, \overline{\sigma_i}} = \{\pref^1_i\} $$
\end{itemize}
Given a profile $\prof{\pref}$ and an order $\prof{\sigma}$, we denote with $app_{\prof{\pref},\prof{\sigma}}(x)$ the approval points of $x$ in $\prof{\pref}$ under the order $\prof{\sigma}$, that is: $$app_{\prof{\pref},\prof{\sigma}}(x)=|\{i\in N\mid x\in A_{p_i,\sigma_i}\}|$$
It is easy to see that $plur_{\prof{\pref}}(x)\leq app_{\prof{\pref},\prof{\sigma}}(x)\leq acc_{\prof{\pref}}(x)$ for any profile~$\prof{\pref}$ and order~$\prof{\sigma}$. 
Moreover, Remark~\ref{rem:power} highlights that our setting allows certain liberties with respect to the approval ballot formed by a voter~$i$. 
Specifically: first, any acceptable subset of alternatives that includes the voter's most preferred one can become the voter's approval ballot under some suitable presentation order; second, given a specific presentation order, any subset of alternatives can become the approval ballot of some voter with a suitable preference; and third, given a specific presentation order, any subset of alternatives that includes the one appearing first in the order can become the approval ballot of some voter with a suitable tolerant preference.

\begin{remark}\label{rem:power}
Let $A\subseteq\alt$ be a subset of alternatives.
Then the following hold:
\begin{itemize}
    \item for every preference $\pref_i \in P(\alt)$ there exists an order $\sigma_i$ such that: $$A_{\pref_i,\sigma_i}=(A\cap \mathit{ACC}(\pref_i))\cup \{\pref^1_i\}$$
    \item for every order $\sigma_i$ there exists a (potentially non-tolerant) preference $\pref_i \in P(\alt)$ such that: $$A_{\pref_i,\sigma_i}=A$$  
    \item for every order $\sigma_i$ there exists a tolerant preference $\pref_i \in P(\alt)$ such that: $$A_{\pref_i,\sigma_i}=A\cup \{\sigma_i^1\}$$ 
\end{itemize}
Thus, presentation orders have a clear power towards ballot formation.
 \end{remark}

A profile $\boldsymbol{A_{\pref,\sigma}}$ is a vector of approval ballots obtained for voters with intrinsic preferences $\boldsymbol{\pref}$ after the presentation of the alternatives following the order $\boldsymbol{\sigma}$:
\[\boldsymbol{A_{\pref,\sigma}}=(A_{\pref_1,\sigma_1},\ldots,A_{\pref_n,\sigma_n})\in\app^N\]
Given an approval profile~$\boldsymbol{A}$, we use $\boldsymbol{A}_{-i}$ to denote the partial profile where the approval ballot of voter $i$ is removed.

\subsection{Approval-based Voting Rules}
An approval-based voting rule $F:\app^N\rightarrow 2^{\alt}$, takes a profile of approval ballots and outputs a set of winning alternatives. For the remainder of the paper, we will simply refer to a \emph{voting rule}, assuming that its inputs are profiles of approval ballots. The canonical such voting rule is standard \emph{approval voting} (\text{AV}), which selects as winners the alternatives with the highest approval count in a profile $\prof{A}=(A_1,\ldots,A_n)$. Concretely:
\[\text{\text{AV}}(\prof{A})= argmax_{x\in \alt}|\{i\in N \mid x\in A_i\}|\]
Another simple voting rule is what we call \emph{nomination} (and can also be thought off as a \emph{quota-1} rule): An alternative is a winner if and only if it is `nominated', i.e., approved by at least one voter.
Formally:
$$\text{Nom}(\prof{A})=\{x\in\alt\mid x\in A_i\}$$

Our main question is thus phrased as follows:
\begin{quote}
    How does the order in which voters examine alternatives influence the outcome of a voting rule?
\end{quote}

Said differently, we examine how the voters' anchoring bias might change the collective outcome since different approval ballots can be produced from the same underlying preference.
Note that in some extreme cases, the formation of an approval ballot is not that affected by the presentation order of the alternatives: for example, this is the case when the voter's preference is intolerant and only her favorite alternative is acceptable and can potentially be approved.
On the other side, when $\sigma_i^1=p^1_i$, the voter encounters her favorite alternative first, which becomes the anchor and no other alternative is approved, making the approval ballot heavily influenced by the presentation order.
This implies Remark~\ref{rem:top-prof}.

\begin{remark} \label{rem:top-prof}
The following observations are in order.
\begin{itemize}
    \item For every preference profile~$\prof{\pref}$ and two orders~$\prof{\sigma}$ and $\prof{\pi}$ such that $\sigma_i^1=\pi_i^1=p^1_i$ for all voters $i\in N$, it holds that $\prof{A_{\pref,\sigma}} =  \prof{A_{\pref, \pi}}$.
    \item For every preference profile~$\prof{\pref}$ such that $t_i = 1$ for all $i\in N$, it holds that $\prof{A_{\pref,\sigma}} =  \prof{A_{\pref, \pi}}$ for all orders $\prof{\sigma}$ and $\prof{\pi}$.
\end{itemize}
We thus have two invariance conditions regarding the presentation orders.
\end{remark}

Remark~\ref{rem:top-prof} obviously provides stability conditions not only for the individual ballots but also for the voting outcome.

\subsection{Anchor-Proofness and Other Properties}

We proceed with the general definition of \emph{anchor-proofness}, i.e., the robustness of voting outcomes with respect to the alternatives' presentation orders.
\begin{defn}
We call a voting rule $F$ \textbf{anchor-proof for  a profile~$\prof{p}$ } if $F(\prof{A_{\pref,\sigma}})=F(\prof{A_{\pref, \pi}})$ for every two orders $\prof{\sigma}, \prof{\pi} \in S(\alt)^n$.
\end{defn}
We call a voting rule $F$ \emph{anchor-proof} if $F$ is anchor-proof for every profile~$\prof{p}$.
Next, we propose a number of other properties (also referred to as \emph{axioms}), which represent desirable features of approval-based voting rules. 
The first one states that the names of the voters should not be considered important in selecting a winning alternative.
\begin{defn}
We say that $F$ is \textbf{anonymous} if for every approval profile $\prof{A} = (A_1,\ldots, A_n) \in \app^N$ and permutation $\lambda:N\rightarrow N$, it holds that $F(A_{\lambda(1)},\ldots,A_{\lambda(n)})=F(\prof{A})$.    
\end{defn}
The next basic axiom states that alternatives should be treated symmetrically by the voting rule.
\begin{defn}
We say that $F$ is \textbf{neutral} if for every  approval profile $\prof{A} = (A_1,\ldots, A_n) \in \app^N$ and permutation $\mu:\alt\rightarrow\alt$ where $\mu(A)=\{\mu(x)\mid x\in A\}$,  it holds that $F(\mu(A_1),\ldots,\mu(A_n))=\mu(F(\prof{A}))$.   
\end{defn}
The following three axioms concern different degrees of consensus that a voting rule should respect. Weak unanimity demands that winners should be selected among the unanimously approved alternatives, if there are any; total unanimity refers to a particular profile where all alternatives are unanimously approved, and suggests that they should all win in that case; unanimity states that winners should be exactly the unanimously approved alternatives, if there are any.

\begin{defn}
We say that $F$ satisfies \textbf{weak unanimity} if for every approval profile~$\prof{A}$ such that $U=\{x\in \alt: x\in A_i \text{ for all } i \in N\} \neq \emptyset$, it holds that $F(\prof{A}) \subseteq U$.    
\end{defn}
\begin{defn}
We say that $F$ satisfies \textbf{total unanimity} if for every approval profile~$\prof{A}$ where $A_i=X$ for all $i\in N$, it holds that $F(\prof{A})=X$.   
\end{defn}
\begin{defn}
We say that $F$ satisfies \textbf{unanimity} if for every approval profile~$\prof{A}$ such that $U=\{x\in \alt: x\in A_i \text{ for all } i \in N\} \neq \emptyset$, it holds that $ F(\prof{A}) = U$. 
\end{defn}

Note that total unanimity is implied by anonymity together with neutrality. Moreover, weak unanimity is a weaker axiom than the Pareto condition proposed by \citet{brandl2022approval} for approval-based rules.
The standard approval voting rule (\text{AV}) satisfies unanimity and thus also satisfies weak and total unanimity. 
The nomination rule (Nom) satisfies total unanimity but violates weak unanimity and unanimity.


\section{Anchor-proofness}\label{sec:anchorproof}

This section investigates questions of \emph{anchor-proofness}: conditions under which the outcome of a voting rule remains invariant to the presentation order of the alternatives.
Our findings can be organized around six sub-questions (summarized in Figure~\ref{fig:summary}). 
Fixing a voting rule, we ask whether it is the case that:
\begin{enumerate}
\item for every ranking profile, the approval winners are independent of the presentation order;
\item there exists a ranking profile for which  the approval winners are independent of the presentation order;
\item there exist two distinct presentation orders that produce the same approval winners for every ranking profile;
\item for every ranking profile there exist two distinct presentation orders that produce the same approval winners;
\item for every two distinct presentation orders there exists a ranking profile for which the approval winners coincide;
\item there exist a ranking profile and two distinct presentation orders for which the approval winners coincide.
\end{enumerate}

Together, these six questions form a complete landscape for assessing qualitative aspects of anchor-proofness: any meaningful inquiry into whether a profile’s winners depend on the presentation order can be stated as one of them.
Clearly, the first question is the strongest and most demanding (and indeed has a negative answer for all non-constant rules), while the last question is the weakest and has a trivial positive answer for all rules. 
The intermediate questions (2–5) are not logically related. 
Note also that questions (4-6) are included mostly for conceptual completeness and are of limited practical interest, since they require information about a concrete preference profile (resp.\ a presentation order) to construct a suitable presentation order (resp.\ preference profile).
The topic of information needed to influence the voting outcome will be revisited in Section~\ref{sec:strategy}.

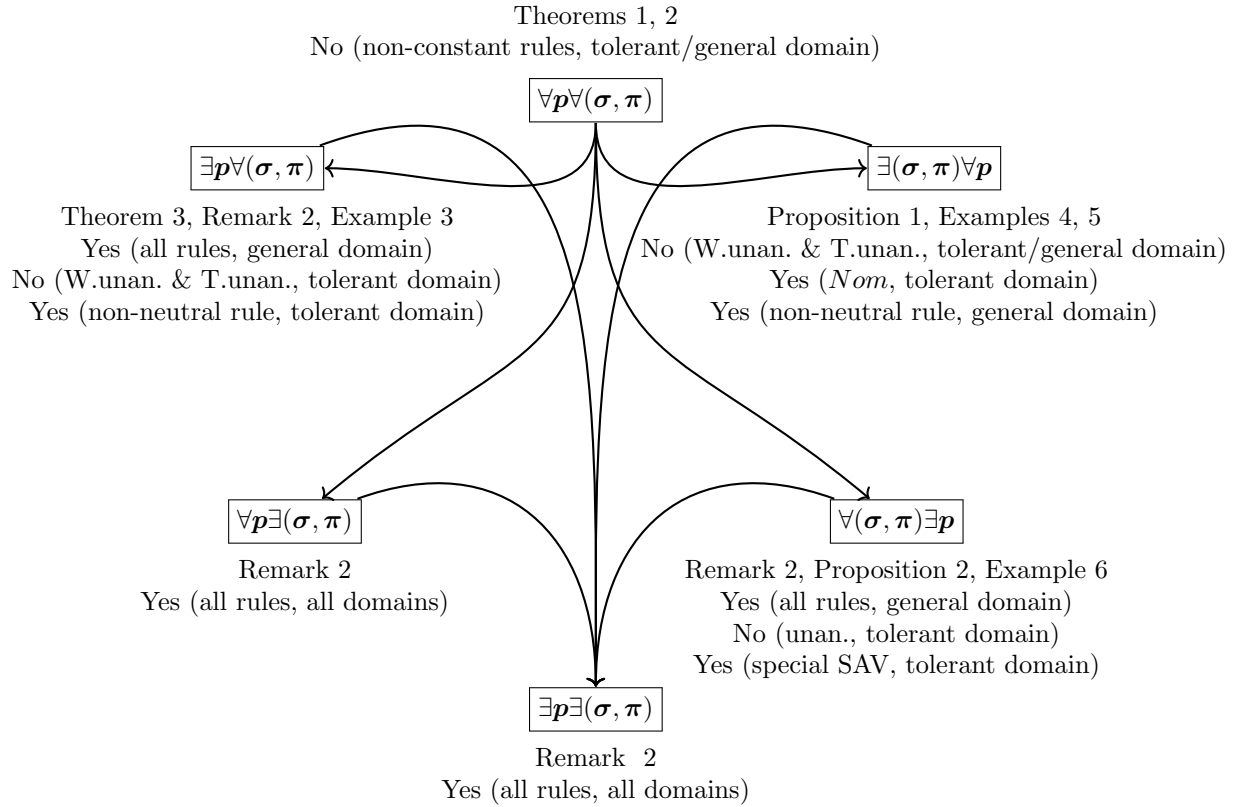
\begin{figure}[htbp]
\centering
\begin{tikzpicture}[
  every node/.style={draw=none, align=center, inner sep=0pt},
]

\useasboundingbox (-6.3, -10.3) rectangle (7.3, 1.7);

\node (top) {\fbox{$\forall \prof{p} \forall (\prof{\sigma},\prof{\pi})$}};
\node[above=0.2cm of top] { Theorem~\ref{thm:mainnonconstant} \\
No (non-constant rules, tolerant/general domain)};

\node (mid1) [below=0.3cm of top, xshift=-4.5cm] {\fbox{$\exists \prof{p} \forall (\prof{\sigma},\prof{\pi})$}};
\node[below=0.2cm of mid1] { Theorem~\ref{thm:imposs-invariant-prof},
Remark~\ref{rem:top-prof}, Example~\ref{ex: a profile for every pair of orders}\\
Yes (all rules, general domain)\\
No (W.unan.\ $\&$ T.unan., tolerant domain)\\
Yes (non-neutral rule, tolerant domain)};

\node (mid2) [below=5cm of top, xshift=-4cm]               {\fbox{$\forall \prof{p} \exists (\prof{\sigma},\prof{\pi})$}};
\node[below=0.2cm of mid2] {  Remark~\ref{rem:top-prof} \\
Yes (all rules, all domains)};

\node (mid3) [below=5cm of top, xshift=4cm]              {\fbox{$\forall (\prof{\sigma},\prof{\pi}) \exists\prof{p}$}};
\node[below=0.2cm of mid3] {Remark~\ref{rem:top-prof}, Proposition~\ref{prop:2inv-prof}, Example~\ref{ex:exrule-poss}\\
Yes (all rules, general domain)\\
No (unan., tolerant domain)\\
Yes (special \text{AV}, tolerant domain)
};

\node (mid4) [below=0.3cm of top, xshift=4.5cm]{\fbox{$ 
 \exists (\prof{\sigma},\prof{\pi}) \forall \prof{p}$}};
\node[below=0.2cm of mid4] { Proposition~\ref{propi:impos-orderforallprofiles}, Examples~\ref{ex:egal-poss}, \ref{ex:non-neutr-nomin} \\
No (W.unan.\ $\&$ T.unan., tolerant/general domain)\\
Yes (Nom, tolerant domain)\\
Yes (non-neutral rule, general domain)
}; 
 
\node (bottom) [below=7.5cm of top] {\fbox{$
\exists \prof{p} \exists (\prof{\sigma},\prof{\pi})$}};
\node[below=0.2cm of bottom] { Remark ~\ref{rem:top-prof}\\
Yes (all rules, all domains)};


\draw[->, thick] (mid1) to[out=-340, in=90, looseness=1.3] (bottom);
\draw[->, thick] (mid2) to[out=-340, in=90, looseness=1.3] (bottom);
\draw[->, thick] (mid3) to[out=-200, in=90, looseness=1.3] (bottom);
\draw[->, thick] (mid4) to[out=-200, in=90, looseness=1.3] (bottom);

\draw[->, thick] (top) to[out=-90, in=0, looseness=1] (mid1);
\draw[->, thick] (top) to[out=-90, in=180, looseness=1] (mid4);
\draw[->, thick] (top) to[out=-90, in=40, looseness=1.4] (mid2);
\draw[->, thick] (top) to[out=-90, in=140, looseness=1.4] (mid3);

\end{tikzpicture}

\caption{Summary of our results. The quantifiers in each frame show the conditions under which we demand that $F(\prof{A_{p,\sigma}}) = F(\prof{A_{p,\pi}})$, always assuming that $\prof{\sigma} \neq \prof{\pi}$. A `Yes' (resp.\ `No') answer means that the corresponding condition can (resp.\ cannot) be satisfied.}
\label{fig:summary}
\end{figure}

\subsection{Anchor-Proofness for All Profiles}\label{sec:anchorproofforallprofile}

We start the presentation of our main result with Theorem~\ref{thm:mainnonconstant}, showing that anchor-proofness is impossible for non-constant voting rules. 
That is, the only rules that are anchor-proof are such that $F(\boldsymbol{A_{\pref,\sigma}}) = F(\boldsymbol{A_{q,\pi}})$ for all preference profiles $\prof{p}$ and $\prof{q}$ and all orders $\prof{\sigma}$ and $\prof{\pi}$.

\begin{teo}\label{thm:mainnonconstant}
Let $F$ be a voting rule that is anchor-proof.
Then, $F$ is constant. 
\end{teo}

\begin{proof}
First, we show that for every non-constant voting rule there exists a voter~$i$, a (partial) approval profile $\prof{A}_{-i}$, and two approval ballots $A_i$ and $A_{i}'$ such that $F(\prof{A}_{-i},A_i)\neq F(\prof{A}_{-i},A'_i)$.
Given $F$, we call this voter a pivotal voter for $\prof{A}_{-i}$.
Indeed, as $F$ is non-constant, there are two different approval profiles $\prof{B}$ and $\prof{C}$ such that $F(\prof{B})\neq F(\prof{C})$.
Starting from the profile $\prof{B}$, in steps for each voter replace $B_i$ with $C_i$ until the outcome of $F(\prof{B})$ changes.
By the hypothesis, this must occur at some point.
Let $i^\star$ be the voter corresponding to this step, and let  $\prof{A}_{-i^\star}$ be the profile such that for all $i<i^\star$, $A_i=B_i$, and for all $i^\star<i$, $A_i=C_i$. 
So $i^\star$ is a pivotal voter for $\prof{A}_{-i^\star}$.

Now we proceed to show that a non-constant rule $F$ cannot be anchor-proof.
Fix such a rule $F$ and let $i^\star$ be a pivotal voter for $\prof{A}_{-i^\star}$. Let also $A_{i^\star}$ and $A'_{i^\star}$ be two approval ballots such that $F(\prof{A}_{-i^\star},A_{i^\star})\neq F(\prof{A}_{-i^\star},A'_{i^\star})$, and such that $A_{i^\star}\cap A'_{i^\star}\neq\emptyset$.
It is always possible to find two such approval ballots with non-empty intersection, as the following reasoning shows. 
Let $B_{i^\star}$ and $B'_{i^\star}$ be two approval ballots such that $F(\prof{A}_{-i^\star},B_{i^\star})\neq F(\prof{A}_{-i^\star},B'_{i^\star})$ but $B_{i^\star}\cap B'_{i^\star}=\emptyset$.
We can consider $C_{i^\star}=B_{i^\star}\cup B'_{i^\star}$.
Then, it holds that either $F(\prof{A}_{-i^\star},C_{i^\star})\neq F(\prof{A}_{-i^\star},B_{i^\star})$ or $F(\prof{A}_{-i^\star},C_{i^\star})\neq F(\prof{A}_{-i^\star},B'_{i^\star})$, because we know that $F(\prof{A}_{-i^\star},B_{i^\star})\neq F(\prof{A}_{-i^\star},B'_{i^\star})$.
Suppose without loss of generality that the latter holds; then, we can choose $A_{i^\star}=C_{i^\star}$ and $A'_{i^\star}=B_{i^\star}$ that fulfill the required conditions.
Now we choose a preference $p_{i^\star}$ for voter $i^\star$  such that there are two different orders $\sigma_i$ and $\pi_i$ with $A_{p_{i^\star},\sigma_i}=A_{i^\star}$ and $A_{p_{i^\star},\pi_i}=A'_{i^\star}$.
This is possible due to the non-empty intersection and Remark~\ref{rem:power}.
Let $\prof{\sigma}$ and $\prof{\pi}$ be two orders such that for all $i\neq i^\star$, $\sigma_i=\pi_i$, and $\prof{p}$ be a preference profile such that $A_{p_{i},\sigma_i}=A_i$.
Then $F(\boldsymbol{A_{\pref,\sigma}}) \neq F(\boldsymbol{A_{p,\pi}})$, contradicting anchor-proofness.
\end{proof}

Thus, only constant rules are robbust to the alternatives' presentation order.
Notably, this holds because the social planner might be able to carefully select an order that takes into account specific information about the voters' preferences, including in particular voters who are pivotal to the collective outcome.
We will see in Section~\ref{sec:strategy} that possibilities arise when such information is not available.

Note that the proof of Theorem~\ref{thm:mainnonconstant} shows that even if there is a single biased voter, whenever there is some instance where this voter is pivotal, then our voting rule will violate anchor-proofness. 
This is always the case for non-constant rules.
Moreover, this result holds even for tolerant profiles; so even for this restricted domain, non-constant voting rules fail to be anchor-proof.

\subsection{Existence of an Anchor-Proof Profile}\label{sec:ExistenceofanAnchorProofProfile}

Having addressed the first question in our list about anchor-proofness, we move on to the second question that asks for, more conservatively,  the existence of \emph{some} profile for which anchor-proofness is achievable. 
By the second point of Remark~\ref{rem:top-prof}, we know that such profiles exist: in particular, every intolerant profile does the job, independently of the voting rule. 
Theorem~\ref{thm:imposs-invariant-prof} shows that this ceases to be possible if we look only within the domain of tolerant profiles and apply a rule satisfying the reasonable axioms of weak and total unanimity.

\begin{teo} \label{thm:imposs-invariant-prof}
Let $F$ be a voting rule that satisfies weak and total unanimity. 
Then, $F$ is not anchor-proof for any tolerant profile~$\prof{\pref}$.
\end{teo}

\begin{proof}
Consider an arbitrary tolerant profile~$\prof{\pref}$ and suppose for contradiction that $F(\prof{A_{\pref,\sigma}})=F(\prof{A_{\pref,\pi}})$ for all orders $\prof{\sigma}, \prof{\pi}$. Then, it must hold that $$F(\prof{A_{\pref,\sigma}})=F(\prof{A_{\pref,\underline{\sigma}}})=F(X,\ldots,X)$$
From total unanimity, this means that $F(\prof{A_{\pref,\sigma}})=\alt$ for all orders $\prof{\sigma}$. 
We distinguish two cases.

First, suppose that all preferences in~$\prof{\pref}$ have the same top alternative~$x$, i.e., $\pref_i^1=x$ for all voters~$i\in N$. 
Then, take the order~$\prof{\sigma}$ to be such that $\sigma_i^1=x$ for all $i\in N$. 
This means that $\prof{A_{\pref,\sigma}}= (\{x\},\ldots,\{x\})$, so weak unanimity applies for some alternative~$y\neq x$ implying that $y\notin F(\prof{A_{\pref,\sigma}})= \alt$. 
We have reached a contradiction.

Second, suppose that there exist two voters $r \neq j$ and two alternatives $x\neq y$ such that $\pref_r^1=x \neq y = \pref_j^1$. 
Then, consider again the order~$\prof{\sigma}$ such that $\sigma_i^1=x$ for all $i\in N$. 
This means that $x\in A_{\pref_i,\sigma_i}$ for all $i\in N$ and $y\notin A_{\pref_r,\sigma_r}=\{x\}$. 
By weak unanimity, we have that $y\notin F(\prof{A_{\pref,\sigma}})= \alt$. 
We have reached a contradiction.
\end{proof}

However, the impossibility of Theorem~\ref{thm:imposs-invariant-prof} for the tolerant domain is not generalizable to arbitrary voting rules: Example~\ref{ex: a profile for every pair of orders} illustrates that there exist non-neutral rules that are anchor-proof for tolerant profiles. 

\begin{exa}\label{ex: a profile for every pair of orders}
Consider the following voting rule $F$. Fix an alternative $x\in \alt$.
Then, for an arbitrary approval profile $\prof{A}=(A_1,\ldots,A_n)$ define:
\[ F(\prof{A}) = \begin{cases}
    \{x\} & \text{ if } x\in A_i \text{ for all } i \in N \\
    \alt & \text{ otherwise }
\end{cases} \]
Now for the profile $\prof{\pref}$ such that $p^1_i=x$ for all $i\in N$ it holds that $F(\prof{A_{\pref,\sigma}})=F(\prof{A_{\pref,\pi}})$ for all orders $\prof{\sigma}$ and $\prof{\pi}$. 
Note that this rule violates neutrality as well as weak and total unanimity.
\end{exa}

As in the previous section, here as well, utilizing crucial information about the voters' preferences by the social planner makes natural voting rules fail anchor-proofness. 
Considering for instance the setting of participatory budgeting, if all projects are so good such that voters accept them all, then there is no escape from the potential influence of a social planner.
However, without such unanimity requirements, we are able to find rules which in certain situations are immune to anchor biases.

\subsection{Order Pairs Preserving Outcomes for All Profiles}\label{sec:OrderPairsPreservingOutcomesforAllProfiles}

We continue with the third question in our list, specifying the existence of a pair of orders producing the same winners for all profiles. This is clearly a weaker requirement than anchor-proofness, which spans across all orders.
Proposition~\ref{propi:impos-orderforallprofiles} (utilizing Lemma~\ref{lem:order-switch}) proves that such a pair of orders does not exist if we use a rule satisfying weak and total unanimity, even if we restrict attention to the domain of all tolerant profiles. Examples~\ref{ex:egal-poss} and~\ref{ex:non-neutr-nomin} show that the impossibility is circumvented in the restricted domain of tolerant profiles or, respectively, in the general domain of all preference profiles, if we consider the nomination rule (which, recall, fails weak unanimity) or, respectively, a rule that violates neutrality. 
Remark~\ref{rem:nomin-nontol} explains that the result for the nomination rule does not extend to the general domain.

\begin{lemma} \label{lem:order-switch}
Consider two orders $\sigma_i$ and $\pi_i$ and two preferences $p_i$ and $p'_i$ such that the following three conditions hold:
\begin{itemize}
    \item[$(i)$] $A_{p_i,\sigma_i} = \alt$;
    \item[$(ii)$] $A_{p_i,\pi_i} \neq \alt$;
     \item[$(iii)$]  $A_{p_i,\pi_i} \subseteq A_{p'_i,\sigma_i}$.
\end{itemize}
Then, we have that $A_{p'_i,\pi_i} = A_{p_i,\pi_i}$.
\end{lemma}

\begin{proof}
Let us first show that $A_{p'_i,\pi_i} \subseteq A_{p_i,\pi_i}$. Aiming for a contradiction, suppose that there exists an alternative $x \in \alt$ such that $x \in A_{p'_i,\pi_i} $ but $x\notin A_{p_i,\pi_i}$. The fact that $x\notin A_{p_i,\pi_i}$ implies the existence of another alternative $a\in A_{p_i,\pi_i}$ such that:
\[ (a,x) \in p_i \quad \text{and} \quad (a,x) \in \pi_i\]
Now, $x \in A_{p'_i,\pi_i} $ and $(a,x) \in \pi_i$ together imply that:
\[(x,a) \in p'_i\]
But by $(a,x) \in p_i$ and knowing from the hypothesis $(i)$ that $A_{p_i,\sigma_i}=\alt$, we have that:
\[(x,a) \in \sigma_i\]
This, together with the fact that $(x,a) \in p'_i$, implies that $a\notin A_{p'_i, \sigma_i} \supseteq A_{p_i,\pi_i}$, which is a contradiction since $a\in A_{p_i,\pi_i}$.

Next, let us show that $A_{p_i,\pi_i} \subseteq A_{p'_i,\pi_i}$. Aiming for a contradiction, suppose that there exists an alternative $y \in \alt$ such that $y \in A_{p_i,\pi_i} \subseteq A_{p'_i,\sigma_i}$ (so also $y \in A_{p'_i,\sigma_i}$),  but $y\notin A_{p'_i,\pi_i} $. Now, from the hypotheses $(i)$ and $(ii)$, it must be the case that any two alternatives $a,b\in A$  are ordered in the same way by $\sigma_i$ and by $\pi_i$: that is, $(a,b) \in \sigma_i$ if and only if $(a,b) \in \pi_i$ (otherwise,  to have that $A_{p_i,\sigma_i} = \alt$, one of the alternatives $a,b$ would not be included in $A_{p_i,\pi_i}$). So, $y \in A_{p'_i,\sigma_i}$ means that:
\[\text{for all } a\in A \text{ such that } (a,y)\in p'_i \text{  it holds that } (y,a) \in \sigma_i\cap \pi_i \]
But since we have that  $y\notin A_{p'_i,\pi_i} $, there must exist another alternative $x\notin A_{p_i,\pi_i}$ such that $(x,y) \in p'_i$ and $x \in A_{p'_i,\pi_i} $ (where $(x,y) \in \pi_i$). This is impossible, since in the first part of this proof we showed that $A_{p'_i,\pi_i} \subseteq A_{p_i,\pi_i}$. We have reached a contradiction and our proof is concluded.
\end{proof}

\begin{propi}\label{propi:impos-orderforallprofiles}
Let $F$ be a voting rule that satisfies weak and total unanimity. 
Then there is no pair of orders $\prof{\sigma}\neq \prof{\pi}$ such that $F(\prof{A_{\pref,\sigma}})= F(\prof{A_{\pref,\pi}})$ for all (tolerant or not) preference profiles $\prof{\pref}$.  
\end{propi}
\begin{proof}
Suppose for contradiction that there are two orders $\prof{\sigma}$ and $\prof{\pi}$ such that $F(\prof{A_{\pref,\sigma}})= F(\prof{A_{\pref,\pi}})$ for all preference profiles $\prof{\pref}$.
We divide the proof into two cases, according to the number $k$ of voters for which the orders are different. That is, $k= |\{i\in N \mid \sigma_i \neq \pi_i\}|$.

\textbf{Case 1:} $k=1$.

We can assume without loss of generality that $\sigma_1 \neq \pi_1$, that is, the orders are different for the first voter but the same for all other voters.
Consider the profile $\prof{\pref}$ such that $\prof{A_{\pref,\sigma}}=(\alt,\ldots,\alt)$ and $\prof{A_{\pref,\pi}}=(A_1,\ldots,\alt)$.
By weak unanimity, we have that $F(A_1,\ldots,\alt)\subseteq A_1$, and by total unanimity $F(\alt,\ldots,\alt)=\alt$.
Thus, $F(A_1,\ldots,\alt)\neq F(\alt,\ldots,\alt)=\alt$, a contradiction with our hypothesis.

\textbf{Case 2:} $k>1$.

We can assume without loss of generality that $\sigma_i\neq \pi_i$ for all $i\in \{1,\ldots,k\}$, that is, the orders are different for the first $k$ voters and the same for all other voters.
Consider the profile $\prof{\pref}$ such that $\prof{A_{\pref,\sigma}}=(\alt,\ldots,\alt)$ and $\prof{A_{\pref,\pi}}=(A_1,\ldots,A_k,\alt,\ldots,\alt)$.
By weak and total unanimity, $\cap_{i=1}^kA_i=\emptyset$, otherwise $F(A_1,\ldots,A_k,\alt,\ldots,\alt)\subseteq \cap_{i=1}^kA_i\neq \alt=F(\alt,\ldots,\alt)$.
By Lemma~\ref{lem:order-switch}, we can create a profile $\prof{\pref}^i$, for $1\leq i\leq k$, such that $\prof{A_{\pref^i,\sigma}}=(\alt,\ldots,\alt,A_i,\alt,\ldots,\alt)$ and $\prof{A_{\pref^i,\pi}}=(A_1,\ldots,A_k,\alt,\ldots,\alt)$.
By weak unanimity we have that $$F(A_1,\ldots,A_k,\alt,\ldots,\alt)=F(\alt,\ldots,\alt,A_i,\alt,\ldots,\alt)\subseteq A_i$$ for all $i=\{1,\ldots,k\}$.
Thus, we have that $$F(A_1,\ldots,A_k,\alt,\ldots,\alt)\subseteq \cap_{i=1}^kA_i=\emptyset,$$ a contradiction.
\end{proof}

\begin{exa} \label{ex:egal-poss}
    Let $F$ be Nom. Then, the two orders $\prof{\sigma}$ and $\prof{\pi}$ defined below are such that $F(\prof{A_{\pref,\sigma}})= F(\prof{A_{\pref,\pi}})$ for all tolerant profiles $\prof{\pref}$. 
    
    Suppose that $n\geq m$ and define $\prof{\sigma} \neq \prof{\pi}$ such that for every alternative $x\in \alt$ it is the case that $x=\sigma^1_i = \pi^1_i$ for some voter $i\in N$. 
    Then, all alternatives are nominated in all profiles $\prof{\pref}$, so $F(\prof{A_{\pref,\sigma}})= F(\prof{A_{\pref,\pi}})=\alt$.

    Second, suppose that $n<m$. 
    Let $\alt'$ be the first $n$ alternatives according to some arbitrary order of $X$. 
    Define $\prof{\sigma} \neq \prof{\pi}$ such that the three following conditions hold:
    \begin{itemize}
        \item 
    for every alternative $x\in \alt'$ it is the case that $x=\sigma^1_i = \pi^1_i$ for some voter $i\in N$; 
        \item 
    all alternatives in $\alt'$ appear before all alternatives in $\alt\setminus \alt'$ for every $i \in N$;
        \item
    all alternatives in $\alt\setminus \alt'$ appear in the same order for every $i \in N$.
    \end{itemize}
    Then, all alternatives in $\alt'$ are nominated in all profiles $\prof{\pref}$, so $\alt' \subseteq F(\prof{A_{\pref,\sigma}})$ and $\alt' \subseteq F(\prof{A_{\pref,\pi}})$. 
    Moreover, given $x\in \alt$, define the following sets:
    \[A^{\sigma}_x = \{y\in \alt : (y,x) \in \sigma_i \text{ for some } i \in N\}\]
    \[A^{\pi}_x = \{y\in \alt : (y,x) \in \pi_i \text{ for some } i \in N\}\]
     Note that $A^{\sigma}_x = A^{\pi}_x$ for all $x \in X\setminus X'$. 
     Now, an alternative $x \in \alt \setminus \alt'$ is a winner according to the nomination rule if and only if there exists a voter $i\in N$ such that $(x,y) \in \pref_i$ for all $y\in A^{\sigma}_x$ (that is, if and only if $x$ is more preferred over all alternatives that are previously presented to a voter).
\end{exa}

Note that Example~\ref{ex:egal-poss} does not work for the general domain (including all non-tolerant profiles). 
In fact, no pair of different orders can guarantee the same outcome for all profiles given the nomination rule (see Remark~\ref{rem:nomin-nontol}).

\begin{remark} \label{rem:nomin-nontol}
Consider arbitrary $\prof{\sigma}\neq \prof{\pi}$. Then, there exist $i\in N$ and $x,y\in X$ such that $(x,y) \in\sigma_i$ and $(y,x) \in \pi_i$. Take the profile~$\prof{\pref}$ such that $x$ is not acceptable for every $j\in N\setminus \{i\}$, while $x$ is acceptable for $i$ and $\pref_i^1=y$ and $\pref_i^2=x$. Then, $x \notin A_{\pref_j,\sigma_j}$ and $x \notin A_{\pref_j,\pi_j}$ for all $j\in N\setminus \{i\}$. Also,  $x \in A_{\pref_i,\sigma_i}$ but $x \notin A_{\pref_i,\pi_i}$. This means that for $F$ being Nom it holds that $x\in F(\prof{A_{\pref,\sigma}}) $ but $x\notin F(\prof{A_{\pref,\pi}})$. 
\end{remark} 

\begin{exa} \label{ex:non-neutr-nomin}
    Consider a non-neutral voting rule according to which the nomination rule is followed except for two alternatives $x$ and $y$ that are excluded from the winning set for all profiles. Then, consider two orders $\prof{\sigma}$ and $\prof{\pi}$ that follow the construction of Example~\ref{ex:egal-poss} with the exception that $x$ and $y$ are shown last for all voters. Then, the reasoning of Example~\ref{ex:egal-poss}  follows. 
\end{exa}    

In contrast to the several impossibilities produced in previous sections, this section sheds light on certain positive observations, as we are able to find two orders for which the winners of a given voting rule are invariant. 
This result becomes relevant for instance in a participatory budgeting scenario where all projects are acceptable to everyone.
A regulatory entity that wants a fair process might then force the social planner to use only the presentation orders that make Nom anchor-proof.
We can see this as a minimal guarantee that the regulatory institution has against ill-intentioned social planners.

\subsection{Profiles Tailored to a Fixed Pair of Orders}\label{sec:Profiles Tailored to a Fixed Pair of Orders}

Lastly, we examine whether for any pair of orders there exists a profile that produces the same winners for both orders. 
The answer is trivial if we allow for any profile: the outcomes of intolerant profiles are invariant for all orders  (recall Remark~\ref{rem:top-prof}). 
However, if we consider a voting rule satisfying unanimity, the answer becomes negative when we only allow for tolerant profiles (Proposition~\ref{prop:2inv-prof}). 
Example~\ref{ex:exrule-poss} shows that a positive answer is achievable for other rules, such as a less decisive variation of \text{AV}.

\begin{propi} \label{prop:2inv-prof}
Consider a voting rule $F$ that satisfies unanimity. Then, there exist two orders $\prof{\sigma}$ and $\prof{\pi}$ such that  $F(\prof{A_{\pref,\sigma}}) \neq F(\prof{A_{\pref,\pi}})$ for all tolerant profiles~$\prof{\pref}$. 
\end{propi}

\begin{proof}
Take two orders $\prof{\sigma}$ and $\prof{\pi}$ such that $\sigma_i^1=x$ and $\pi_i^1=y$, for alternatives $x\neq y$ and  for all $i\in N$. 
Then, suppose for contradiction that  $F(\prof{A_{\pref,\sigma}}) = F(\prof{A_{\pref,\pi}})$ for some tolerant profile $\prof{\pref}$.
    By construction, we know that $x \in A_{\pref_i, \sigma_i}$ and $y \in A_{\pref_i,\pi_i}$ for all $i\in N$.  
    Thus, unanimity implies that $\{x,y\}\subseteq F(\prof{A_{\pref,\sigma}}) = F(\prof{A_{\pref,\pi}})$. 
    Without loss of generality, suppose that $(x,y) \in \pref_i$ for some $i\in N$ (the case where $(y,x) \in \pref_i$ is completely symmetric). Then, since $\pi_i^1=y$, we have that $x \notin A_{\pref_i,\pi_i}$. So by unanimity, $x\notin F(\prof{A_{\pref,\pi}})$. Contradiction.
\end{proof}

\begin{exa} \label{ex:exrule-poss}
Consider the voting rule $F$ that returns all alternatives when at least one voter approves at least two alternatives and behaves as \text{AV} otherwise (when each voter approves only one alternative). 
We can find a profile~$\prof{\pref}$ that is anchor-proof for $F$.
Indeed, for arbitrary orders $\prof{\sigma}\neq \prof{\pi}$ consider a profile such that $\pref_1=(\sigma_1^2, \sigma_1^1,\ldots)$ and  $\pref_2=(\pi_1^2, \pi_1^1,\ldots)$. It will be the case that $A_{\pref_1, \sigma_1} = \{\sigma_1^1, \sigma_1^2\}$ and $A_{\pref_2, \pi_2} = \{\pi_2^1, \pi_2^2\}$. Thus, $F(\prof{A_{\pref,\sigma}})= F(\prof{A_{\pref,\pi}})= \alt$ according to the definition of the rule.
\end{exa}

This section speaks about scenarios where the presentation orders are fixed and limited, for instance within electoral contexts that are supposed to exhibit the alternatives to voters in a pre-determined way. We ask whether there are specific preference structures that would induce robust outcomes for at least  two different fixed presentation orders within groups of tolerant voters, thus restricting the potential influence of the social planner. We find that this rather weak requirement is not satisfied as long as the voting rule adheres to unanimity. However, there do exist situations where the social planner has no influence over the outcome. One of them concerns the tolerance of voters: In particular, consider scenarios where projects in participatory budgeting are extreme, in the sense that each project clearly benefits one part of the society while harming all the rest; then, voters can be assumed to accept only those projects from which they benefit and to reject all others, yielding an intolerant preference profile in which the outcome is independent of the presentation order.
The other situation that sets constraints on the influence of the social planner does not necessarily demand intolerant voters, but employs a voting rule  that is only meaningful in situations of indirect intolerance: unless each voter only approves a single alternative where simple \text{AV} may be applied, that rule returns a tie between all alternatives and is thus completely undecisive.


\section{Strategic Presentation Orders} \label{sec:strategy}

 In many group decision environments, a social planner---such as an electoral authority, participatory-budgeting platform, or other decision designer---controls precisely how alternatives are presented to voters. This raises a strategic question. Suppose that the planner has her own preferences over collective outcomes. Can she exploit the presentation order to induce a desirable result, when voters' ballots are subject to an anchoring bias, even if their intrinsic preferences are fixed? The answer depends not only on the voting rule but also on the information available to the planner about voters' preferences.

We first study the planner's ability to manipulate the outcome when she has \emph{full information} about the voters' preferences. The strategic problem is closely connected to our notion of anchor-proofness. If a voting rule is anchor-proof, the planner cannot affect the collective outcome through the presentation order, regardless of her preferences or information. For non-anchor-proof rules, however, the planner may be able to exploit order dependence. This underlies our characterizations of anchor-proof profiles under \text{AV} and \text{Nom} (Propositions~\ref{prop:troubledwaters} and~\ref{prop:nomination-char}), as well as our general result for all weakly unanimous rules (Proposition~\ref{prop:weakunacharacterisationprofile}).

We then relax the assumption of full information and study precisely how the planner's ability to manipulate the outcome depends on the information she has about voters' preferences. We distinguish between various degrees of information. Under \emph{zero information}, the planner must choose a presentation order without observing any information about voters' preferences. We first show that, under zero information, manipulation is impossible for any rule satisfying weak and total unanimity, including \text{AV} (Theorem~\ref{thm:zero-info}), and also for the nomination rule (Proposition~\ref{prop:nom-info0manip}). We then demonstrate that even limited aggregate information---such as acceptability points or plurality information---can restore manipulability for broad classes of rules (Theorems~\ref{thm:manip-unanimous} and~\ref{thm:manip-weakandtotalunaniinfoplu}, and Proposition~\ref{prop:nom-manip}).

\subsection{Anchor-proof Profiles under Full Information}

Let us start with our results regarding anchor-proof profiles. At such profiles, the planner's control over the presentation order does not provide an additional instrument for influencing the collective outcome. Under \text{AV}, homogeneity seems to help with anchor-proofness, but to an extent, since profiles with  more than one unanimously accepted alternative are not anchor-proof.

\begin{propi}\label{prop:necessaryforinvariance}
$\text{AV}$ is anchor-proof for a profile $\prof{\pref}=(\pref_1,\ldots,\pref_n)$ only if $\prof{\pref}$ does not have more than one unanimously accepted alternative; that is, only if $|\{x\in \alt \mid  x\in \mathit{ACC}(\pref_i) \text{ for all } i \in N \} | \leq 1$.
\end{propi}

\begin{proof}
Suppose that the profile $\prof{\pref}$ has at least two unanimously accepted alternatives $x$ and $y$. Without loss of generality, suppose that there exists a voter $j\in N$ such that $(x, y) \in \pref_j$, meaning that voter $j$ prefers $x$ to $y$.
Using the order $\prof{\sigma}$ such that $\sigma^1_i=x$ for every $i\in N$, the social planner can make sure that $x$ gets $n$ approval points while all other alternatives (including $y$) get fewer approval points, and thus $x$ is the unique \text{AV} winner. Alternatively, using $\prof{\pi}$ such that $\pi^1_i=y$ for every $i\in N$, the social planner can ensure that $y$ gets $n$ approval points and is thus an \text{AV} winner. 
So $\prof{\pref}$ is not anchor-proof.
\end{proof}

 The idea behind anchor-proofness of AV for certain profiles is that when plurality leaders are sufficiently dominant, then no alternative can overtake them via acceptability votes induced by particular presentation orders.

\begin{propi}\label{prop:troubledwaters}
$\text{AV}$ is anchor-proof for a profile $\prof{\pref}$ if and only if:
$$ \text{ for all } x\in argmax_{x\in X} plur_{\prof{\pref}}(x) \text{ it holds that } plur_{\prof{\pref}}(x) \geq acc_{\prof{\pref}}(y) \text{ for all } y\in X,$$
where the equality only holds exactly when $y\in argmax_{x\in X} plur_{\prof{\pref}}(x)$. 
\end{propi}

\begin{proof}
Let $F$ be $\text{AV}$.

(only if) Suppose that there is an alternative  $y\notin argmax_{x\in X}plur_{\prof{\pref}}(x)$ not plurality-leading  such that for some plurality leader $x\in argmax_{x\in X}plur_{\prof{\pref}}(x)$ it is the case that $plur_{\prof{\pref}}(x) \leq acc_{\prof{\pref}}(y)$.
Note that $F(\prof{A_{\prof{\pref},\overline{\sigma}}})=argmax_{x\in X}plur_{\prof{\pref}}(x)$.
By Remark~\ref{rem:power}, there exists an order $\prof{\mu}$ such that for every voter $i$ who finds $y$ acceptable, it holds that $A_{p_i,\mu_i}=\{p^1_i,y\}$; and $A_{p_k,\mu_k}=p^1_k$ otherwise.
Therefore $y\in F(\prof{A_{\prof{\pref},\mu}})\neq argmax_{x\in X}plur_{\prof{\pref}}(x)$, which means that $\prof{\pref}$ is not anchor-proof. 
The proof is similar if  $y\in argmax_{x\in X}plur_{\prof{\pref}}(x)$.

(if) For any presentation order, each plurality leader $x\in argmax_{x\in X}plur_{\prof{\pref}}(x)$ gets an amount of approval points at least equal to its plurality points (recall that every alternative at the top of a preference ranking will be approved by that voter, independently of the presentation order).
Every other alternative gets an amount of approval points at most equal to its acceptability score, which by hypothesis is no larger than the plurality score of a leader, with equality only for plurality leaders. Therefore the winners are always  the plurality leaders.
\end{proof}

Proposition~\ref{prop:nomination-char} answers the same question for Nom. 
Intuitively, anchor-proofness holds precisely when all acceptable alternatives are already someone's top choice; thus reordering cannot unlock new winners.

\begin{propi}
\label{prop:nomination-char}
Nom is anchor-proof for a
 profile $\prof{\pref}$ 
if and only if  $\mathit{PLUR}(\prof{\pref}) = \mathit{ACC}(\prof{\pref})$.
\end{propi}

\begin{proof}

Let $F$ be Nom. We prove the two directions.

(only if) Suppose that   $\mathit{PLUR}(\prof{\pref}) \neq \mathit{ACC}(\prof{\pref})$. 
Then, there exist an alternative $x\in \alt$ and a voter $i\in N$ such that $x \in \mathit{ACC}(\pref_i)$ and $x\neq \pref_i^1$. 
Consider the order $\prof{\overline{\sigma}}$ and the different order $\prof{\pi}$ such that $\pi_j^1 = \pref_j^1$ for all $j\in N \setminus \{i\}$ and $\pi_i^1 = x$ and $\pi_i^2 = \pref_i^1$. 
Then, by definition of the nomination rule we have that $F(\prof{A_{\pref,\overline{\sigma}}}) = \mathit{PLUR}(\prof{\pref}) \neq \mathit{PLUR}(\prof{\pref}) \cup \{x\} = F(\prof{A_{\pref,\pi}})$. 

(if) Suppose that $\mathit{PLUR}(\prof{\pref}) = \mathit{ACC}(\prof{\pref})$. For every $x\in \mathit{PLUR}(\prof{\pref})$ and order $\prof{\sigma}$ we know that $x\in A_{\pref_i,\sigma_i}$ for some voter~$i$. 
Then, by definition of the nomination rule we have that  $\mathit{PLUR}(\prof{\pref}) \subseteq F(\prof{A_{\pref,\sigma}}) \subseteq \mathit{ACC}(\prof{\pref}) = \mathit{PLUR}(\prof{\pref})$ for every order $\prof{\sigma}$, which means that $F(\prof{A_{\pref,\sigma}}) = \mathit{PLUR}(\prof{\pref})$ for every order $\prof{\sigma}$.  
\end{proof}

Propositions~\ref{prop:troubledwaters} and~\ref{prop:nomination-char} characterize the anchor-proof profiles for two particular rules, namely AV and Nom.
Next, Proposition~\ref{prop:weakunacharacterisationprofile} characterizes the profiles that are anchor-proof for \emph{all} weakly unanimous rules.\footnote{We may attempt to characterize the profiles for which  an arbitrary weakly unanimous rule is anchor-proof.
But since any rule can be transformed into a weakly unanimous rule by adding the condition to select the set of unanimously approved alternatives when this set is non-empty, anchor-proofness is strongly rule-dependent.} Recall that an alternative $x \in \alt$ is called \emph{unanimously accepted} in a profile $\prof{\pref} = (\pref_1,\ldots,\pref_n)$ if $x \in \mathit{ACC}(\pref_i)$ for every $i\in N$.

\begin{propi}\label{prop:weakunacharacterisationprofile} 
 It holds that every weakly unanimous voting rule is anchor-proof for a profile~$\prof{\pref}$ if and only if $\prof{p}$ is intolerant or it has a unique unanimously accepted alternative that is ranked first by every voter.
\end{propi}

\begin{proof}
Given an approval profile $\prof{A}=(A_1,\ldots, A_n)$, let us write $U(\prof{A}) = \{x\in \alt \mid x\in A_i \text{ for every } i\in N\}$ for the set of unanimously approved alternatives. We now proceed to the proof.

(only if) Let $\prof{p}$ be  a profile for which every weakly unanimous rule is anchor-proof, and suppose for contradiction that $\prof{\pref}$ is not intolerant and it does not have a unique unanimously accepted alternative that is ranked first by every voter.
We will consider three cases, and on each of these three cases we will present a weakly unanimous voting rule $F$ that is not anchor-proof for $\prof{\pref}$.

\textbf{Case 1:} $\prof{p}$ is a profile with more than one unanimously accepted alternative.

By Proposition~\ref{prop:necessaryforinvariance}, we know that \text{AV} is not anchor-proof for~$\prof{p}$.

\textbf{Case 2:} $\prof{p}$ is a profile with a unique unanimously accepted alternative $x$, but $x$ is not ranked first by every voter.

We define the weakly unanimous rule $F$ such that
 $F(\prof{A})=U(\prof{A})$  if $U(\prof{A})\neq\emptyset$, and $F(\prof{A})=\alt$, otherwise.
Let $\prof{\sigma}$ be an order such that $\sigma_i^1=x$ for all $i\in N$.
Then we have that $F(\prof{A_{p,\sigma}})=\{x\}\neq \alt = F(\prof{A_{p,\underline{\sigma}}})$.

\textbf{Case 3:} $\prof{p}$ is a profile without a unanimously accepted alternative.

We define the weakly unanimous rule $F$ such that $F(\prof{A})=U(\prof{A})$ if $U(\prof{A})\neq\emptyset$, and otherwise $F(\prof{A})$ selects the largest approval ballot in $\prof{A}$ (in case of a tie, it selects the approval ballot of the voter with the minimal index).
Then for the profile $\prof{\pref}$ we know that $U(\prof{A_{\pref,\sigma}}) =\emptyset$ for all orders $\prof{\sigma}$. So we have that $|F(\prof{A_{p,\overline{\sigma}}})|=1$ (because under $\prof{\overline{\sigma}}$ all approval ballots are of size 1) and $|F(\prof{A_{p,\underline{\sigma}}})|>1$ (because the profile is not intolerant, so under $\prof{\underline{\sigma}}$ there exists at least one approval ballot of size larger than 1).

(if)
Let $\prof{p}$ be an intolerant profile. Then by Remark~\ref{rem:top-prof}, we have that every approval-based rule is anchor-proof for $\prof{p}$

Now let $\prof{p}$ be a profile with a unique unanimously accepted alternative $x$ that is ranked first by every voter.
Then, for all $i\in N$ and all $\sigma_i \in S(\alt)$, it is the case that $x\in A_{p_i,\sigma_i}$.
Thus we have that $U=\{x\}$, and by weak unanimity it holds that $F(\prof{A_{p,\sigma}})=\{x\}$ for all orders $\prof{\sigma}$, which means that $F$ is anchor-proof for $\prof{\pref}$.
\end{proof}

It should be noted that when adding total unanimity to the requirements of Proposition~\ref{prop:weakunacharacterisationprofile}, the result remains the same. 
This is because total unanimity only plays a role in the particular profile where every approval ballot is exactly~$\alt$.
The profiles employed in the characterization already omit this profile.

Together, the results of this section underline that although no non-constant approval-based rule is anchor-proof on the unrestricted domain, order invariance can arise at particular profiles. 

\subsection{Anchor-proofness under Partial Information}\label{sec:anchproofunderpartialinfo}

So far we have been able to find only a small number of profiles for which particular rules are anchor-proof.
Crucially, these results rest on a strong assumption about the information available to the planner: anchor-proofness often fails precisely because the planner might be able to select a specific presentation order relying on the intrinsic preferences of the voters.

Next, we turn to settings where the planner lacks full information about the voters’ preferences. Using the model of \citet{ReijngoudEndrissAAMAS2012}, we define a notion called `information function' $\info$ that maps each preference profile $\prof{\pref}$ to a type of information. 
We will study the following such functions:
\begin{itemize}
    \item \textbf{Zero:} $\info_{0}(\prof{\pref})$ returns a constant value, independently of $\prof{\pref}$; that is, it provides no useful information. 
    \item \textbf{Acceptability points:} $\info_{acc}(\prof{\pref})$ returns, for each alternative, the number of voters that consider it acceptable in $\prof{\pref}$. 
    \item \textbf{Acceptability sets:} $\info_{\{acc\}}(\prof{\pref})$ returns, for each alternative, the set of voters that consider it acceptable in $\prof{\pref}$. 
   \item \textbf{Plurality points:} $\info_{pl}(\prof{\pref})$ returns, for each alternative, the number of voters that rank it on top in $\prof{\pref}$. 
    \item \textbf{Plurality sets:} $\info_{\{pl\}}(\prof{\pref})$ returns, for each alternative, the set of voters that rank it on top in $\prof{\pref}$.
       \item \textbf{Full:} $\info_1(\prof{\pref})$ returns precisely the profile $\prof{\pref}$. 
\end{itemize}

Each information function gives rise to a set of profiles that the social planner considers possible, consistently with the provided information. Formally, the set of profiles that are considered possible given a certain profile $\prof{\pref} \in P(\alt)^N$ is defined as follows:
\[W_{\info(\prof{\pref})} = \{\prof{\pref'} \in P(\alt)^N \mid \info(\prof{\pref}) = \info(\prof{\pref'}) \} \]
For example, $W_{\info_0(\prof{\pref})} = P(\alt)^N$  and  \mbox{$W_{\info_1(\prof{\pref})}= \{\prof{\pref}\}$} for any profile~$\prof{\pref}$.
Next, Definition~\ref{def:info} formalizes the comparison between information functions in terms of the degree of information they provide.

\begin{defn} \label{def:info}
Let $\info$ and $\info'$ be two information functions. 
We call $\info$ \textbf{at least as informative as} $\info'$ if for every profile $\prof{\pref}$ it holds that $W_{\info(\prof{\pref})} \subseteq W_{\info'(\prof{\pref})}$. 
\end{defn}

Considering the information functions that we presented, $\info_0$ (respectively,  $\info_1$) is the least (respectively, the most) informative of all; $\info_{acc}$ is less informative than $\info_{\{acc\}}$; and 
$\info_{pl}$ is less informative than $\info_{\{pl\}}$.
Note that $\info_{pl}$ may at first appear at least as informative as  $\info_{acc}$, but
technically this is not true since $\info_{pl}$ restricts the alternatives that can be on the top of the profile, while $\info_{acc}$ restricts the acceptance thresholds of the voters.

Now, we assume that the social planner has a preference ranking $\succ_{sp}$ over all the possible outcomes of the aggregation process (that is, $\succ_{sp}\in L(2^{\alt})$).
Thus, she might have the incentive to show the alternatives in a specific order (which we call a \emph{strategy}) to obtain a more preferred outcome.

We ask whether, under partial information, the social planner can manipulate the voting outcome to her advantage.
Given a set of possible profiles $W\subseteq P(\alt)^N$ consistent with the information of the social planner and a voting rule, we call a strategy \emph{optimal} for the social planner if $(i)$ for every possible profile in $W$ it induces an outcome that is at least as good as the outcome of every other strategy, and $(ii)$ in some possible profile in $W$ it induces an outcome that is strictly better than the outcome of some other strategy.\footnote{Case $(ii)$ is important to avoid the scenario where all strategies give the same outcome in all possible profiles, meaning that none of them can be considered optimal.} If the social planner has an optimal strategy given her information, then she can manipulate the voting rule. Definition~\ref{def:manip-info} formalizes this idea.

\begin{defn} \label{def:manip-info}
    Let $\info$ be an information function and $F$ be a voting rule. We say that $F$ is \textbf{manipulable under $\boldsymbol{\info}$ on a profile~$\prof{\pref}$ by a social planner with preference $\succ_{sp}$} if there exists a strategy $\prof{\sigma^\ast} \in S(\alt)$ for which the following two conditions hold:
    \begin{itemize}
        \item for every profile $\prof{\pref'}\in W_{\info(\prof{\pref})}$, and for every order $\prof{\sigma}\neq\prof{\sigma^\ast}$, it is the case that $F(\prof{A_{\pref',\sigma^\ast}})\succeq_{sp} F(\prof{A_{\pref',\sigma}})$;
        \item for some profile $\prof{\pref'}\in W_{\info(\prof{\pref})}$, and for some order $\prof{\sigma}\neq\prof{\sigma^\ast}$, it is the case that $F(\prof{A_{\pref',\sigma^\ast}})\succ_{sp} F(\prof{A_{\pref',\sigma}})$. 
    \end{itemize} 
\end{defn}

The planner's instrument is therefore the presentation order rather than the voters' intrinsic preferences or the voting rule itself. By changing the order in which alternatives are encountered, the planner changes the approval ballots generated by the anchoring mechanism and thereby potentially changes the collective outcome.
We next define the general notion of manipulability for a voting rule~$F$.

\begin{defn} \label{def:manip}
  Let $\info$ be an information function and $F$ be a voting rule. We say that $F$ is \textbf{manipulable under} $\boldsymbol{\info}$ if it is manipulable on some profile~$\prof{\pref}$ by a social planner with some preference $\succ_{sp} \in L(2^{\alt})$. 
\end{defn}

 Definition~\ref{def:manip} is a direct analog of the definition of manipulability in voting under partial information by \cite{ReijngoudEndrissAAMAS2012}, which has been adopted in numerous other works since its original conception \citep[see, for example,][]{endriss2016strategic,Veselova2020,GORI2021manipunderincominfo}. It is a rather demanding notion, according to which a strategy will only be used if it is `safe', meaning that the planner will never risk obtaining a worse outcome in some possible scenario. The strength of this condition increases the relevance of the manipulability results that we will soon present.
We start with an intuitive observation: an increase on the degree of the planner's information benefits her ability to manipulate.

\begin{lemma} \label{lem:informative}
Take an information function $\info$ that is at least as informative as an information function $\info'$. Then, every voting rule $F$ that is manipulable under $\info'$ is also manipulable under $\info$. 
\end{lemma}

\begin{proof}
    Since $F$ is manipulable under $\info'$ for a planner's preference $\succ_{sp}$, an optimal order $\prof{\sigma^\ast}$, and a profile $\prof{\pref}$, we know the following from the second condition of manipulability:
\[F(\prof{A_{\pref',\sigma^\ast}})\succ_{sp} F(\prof{A_{\pref',\sigma}}) \text{ for some order } \prof{\sigma}\neq\prof{\sigma^\ast} \text{ and profile } \prof{\pref'} \in W_{\info'(\prof{\pref})}\]  

  We will show that $F$ is manipulable under $\info$ for the same planner's preference $\succ_{sp}$ and optimal order $\prof{\sigma^*}$, and for the profile $\prof{\pref'}$.  
    Note that $\prof{\pref'} \in W_{\info(\prof{\pref'})}$ by definition. So the second condition of manipulability is also satisfied for $\info$ because of the following:
\[F(\prof{A_{\pref',\sigma^\ast}})\succ_{sp} F(\prof{A_{\pref',\sigma}}) \text{ for some order } \prof{\sigma}\neq\prof{\sigma^\ast} \text{ and profile } \prof{\pref'} \in W_{\info(\prof{\pref'})}\]     
 From the first condition of manipulability of $\info'$, the following  holds:
\[F(\prof{A_{\pref'',\sigma^\ast}})\succeq_{sp} F(\prof{A_{\pref'',\sigma}}) \text{ for all  orders } \prof{\sigma}\neq\prof{\sigma^\ast} \text{ and profiles } \prof{\pref''} \in W_{\info'(\prof{\pref})}\]
In addition, since $\info$ is at least as informative as $\info'$, we know that $W_{\info(\prof{\pref})} \subseteq W_{\info'(\prof{\pref})}$ for all profiles $\prof{\pref}$. Specifically, it holds that  $W_{\info(\prof{\pref'})} \subseteq W_{\info'(\prof{\pref'})}$. Moreover, the fact that $\prof{\pref'} \in W_{\info'(\prof{\pref})}$ implies that  $W_{\info'(\prof{\pref})} = W_{\info'(\prof{\pref'})} $ by definition.
So the first condition of manipulability is also satisfied for $\info$ because of the following:
\[F(\prof{A_{\pref'',\sigma^\ast}})\succeq_{sp} F(\prof{A_{\pref'',\sigma}}) \text{  for all  orders } \prof{\sigma}\neq\prof{\sigma^\ast} \text { and profiles } \prof{\pref''} \in W_{\info(\prof{\pref'})}\] 
Our proof is concluded.
\end{proof}

We first consider the extreme case in which the planner has no information about the voters' preferences. The planner must therefore choose the presentation orders without conditioning them on the realized preference profile. 

Theorem~\ref{thm:zero-info} shows that the planner's informational disadvantage is decisive. Without any information about voters' preferences, the planner cannot safely exploit the anchoring mechanism when the voting rule satisfies weak and total unanimity. The nomination rule satisfies the same zero-information robustness, as later shown in Proposition~\ref{prop:nom-info0manip}.

\begin{teo}\label{thm:zero-info}
Let $F$ be a voting rule that satisfies weak and total unanimity. Then, $F$ is not manipulable under $\info_{0}$.
\end{teo}
\begin{proof}
Aiming for a contradiction, suppose that $\prof{\sigma^\ast}$ is an optimal strategy for a social planner with preference $\succ_{sp}$, who has zero information.
Now consider the tolerant profile $\prof{\pref}$ such that $A_{\pref_i,\sigma^\ast_i}=\alt$ for all $i\in N$ (this can be done by Remark~\ref{rem:power}).
By total unanimity we have that $F(\prof{A_{\pref,\sigma^\ast}})=\alt$.
Fix a voter~$j$, and let $\prof{\sigma}\neq\prof{\sigma^\ast}$ be an order such that $\sigma_i=\sigma_i^{\ast}$ for all $i\in N\setminus \{j\}$, where $\sigma_j$ and $\sigma_j^{\ast}$ differ only in the last two showed alternatives. 
Thus, we have that $A_{\pref_i,\sigma_i}=\alt$ for all $i\in N\setminus \{j\}$, and $A_{\pref_j,\sigma_j}=\alt\setminus \{x\}$ for some alternative $x$.
By weak unanimity, we have that $F(\prof{A_{\pref,\sigma}})=A\subseteq\alt\setminus \{x\}$.
By optimality of $\prof{\sigma^\ast}$, it is the case that $\alt\succ_{sp}A$.

Since $\sigma_i^{\ast1}\in A_{p_i,\sigma_i}$ for all $i\in N$ and $m>2$, it is possible, by Remark~\ref{rem:power}, to find a profile $\prof{\pref'}$ such that $\prof{A_{\pref',\sigma^\ast}}=\prof{A_{\pref,\sigma}}$.
Thus, $F(\prof{A_{\pref',\sigma^\ast}})=A$.
Again, by Remark~\ref{rem:power}, we can find an order $\prof{\sigma'}$ such that $A_{\pref'_i,\sigma'_i}=\alt$ for all $i\in N$, and by total unanimity, $F(\prof{A_{\pref',\sigma'}})=\alt$.
By optimality of $\prof{\sigma^\ast}$, we have that $A\succ_{sp}\alt$.
Thus, we reach a contradiction.
\end{proof}

Since \text{AV} satisfies weak and total unanimity, the following holds:

\begin{coro}
    \text{AV} is not manipulable under $\info_0$.
\end{coro}

The zero-information result raises the question of how much information is required for manipulation to become possible. We next show that full information is not necessary: even coarse aggregate information about voters' preferences can restore the planner's ability to manipulate the collective outcome. Theorem~\ref{thm:manip-unanimous} specifically concerns all weakly unanimous rules that are also anonymous and neutral. 
Note that if we do not consider anonymity and neutrality, we need to include some ad-hoc condition ensuring that the rule produces at least two different outcomes within some interesting sets of possible profiles.

\begin{teo} \label{thm:manip-unanimous}
    Let $F$ be a voting rule that satisfies weak unanimity, anonymity, and neutrality. 
    Then, $F$ is manipulable under $\info_{acc}$.
\end{teo}

\begin{proof}
    Consider a profile~$\prof{\pref}$ where a unique alternative $a$ has $n$ acceptability points and all other alternatives have $n-1$ acceptability points. Take a social planner with preference $\succ_{sp}$ that ranks $\{a\}$ first. Then, the order $\prof{\sigma}$ such that $\sigma_i^1 = a$ for all $i\in N$ is an optimal strategy for the social planner. 
    Indeed, the first condition of manipulability is satisfied because $F(\prof{A_{\pref,\sigma}}) = \{a\} \succeq_{sp} X'$, for every $X'\subseteq \alt$ that can be the outcome of the that profile under another order, as $a$ is selected by weak unanimity.
    For the second condition of manipulability, we need to find a profile with the same acceptability points as $\prof{\pref}$ for all alternatives, in which some order produces an outcome different from $\{a\}$. For this purpose, consider a profile $\prof{\pref'}$ together with an order $\prof{\sigma'}$ such that each alternative is excluded by exactly one approval ballot.
 Then, anonymity and neutrality imply that $F(\prof{A_{\pref',\sigma'}}) = \alt$, and we are done.
\end{proof}

\begin{teo}\label{thm:manip-weakandtotalunaniinfoplu}
    Let $F$ be a voting rule that satisfies weak and total unanimity. 
    Then, $F$ is manipulable under $\info_{pl}$.
\end{teo}

\begin{proof}
    Consider a profile~$\prof{\pref}$ where a unique alternative $a$ has $n$ plurality points. 
    Take a social planner with preference $\succ_{sp}$ that ranks $\{a\}$ first. 
    Then, the order~$\prof{\sigma}$ such that $\sigma_i^1 = a$ for all $i\in N$ is an optimal strategy for the social planner. 
    Indeed, the first condition of manipulability is satisfied because $F(\prof{A_{\pref,\sigma}}) = \{a\}$ by weak unanimity and $ \{a\} \succeq_{sp} X'$, for every $X'\subseteq \alt$ that can be the outcome of that profile $\prof{\pref}$ under another order.
    For the second condition of manipulability, we need to find a profile with the same plurality points as $\prof{\pref}$ for all alternatives, in which some order produces an outcome different from $\{a\}$. 
    For this purpose, consider a tolerant profile $\prof{\pref'}$ where all voters hold the same preferences ranking $a$ first and consider the order $\underline{\prof{\sigma}}$.
    Then, by total unanimity, $F(\prof{A_{\pref',\underline{\prof{\sigma}}}}) = X$, and we are done.
\end{proof}

Corollary~\ref{cor:SAV-manip} holds because \text{AV} satisfies the conditions of Theorems~\ref{thm:manip-unanimous} and~\ref{thm:manip-weakandtotalunaniinfoplu}, and Lemma~\ref{lem:informative} implies that manipulability pursues as long as the information of the planner increases.

\begin{coro} \label{cor:SAV-manip}
 $\text{AV}$ is manipulable under $\info_{acc}$ and $\info_{pl}$ (as well as under $\info_{\{ acc\}}$ and  $\info_{\{pl\}}$).
\end{coro}

We continue with exploring the degree of information needed for the nomination rule to become manipulable.
First, we show that when no information is available, Nom cannot be manipulated.

\begin{propi} \label{prop:nom-info0manip}
Nom is not manipulable under $\info_{0}$.   
\end{propi}

\begin{proof}
Aiming for a contradiction, suppose that $\prof{\sigma^\ast}$ is an optimal strategy for a social planner with preference $\succ_{sp}$, who has zero information.
Wlog, suppose that $\sigma^\ast_i=(x,y,z,\ldots)$.
Consider the profile $\prof{p}$ such that all the voters but $i$ accept only alternative $x$, while voter $i$ ranks $z$ first, $y$ second, and $x$ third; and accepts only those three alternatives.
Then $\text{Nom}(\prof{A_{\pref,\sigma^\ast}})=\{x,y,z\}$, and by optimality, $\{x,y\}\succeq_{sp} X'$ for every $X'\subseteq \alt$.
In particular, for $X'=\{x,z\}$, that is the outcome obtained when $\sigma_i=(x,z,\ldots)$ is used.
Now consider $\prof{p'}$ such that all the voters but $i$ accept only alternative $x$, while voter $i$ ranks $z$ first, $x$ second, and $y$ third, and accepts only those three alternatives.
Then $\text{Nom}(\prof{A_{\pref,\sigma^\ast}})=\{x,y\}$, and by optimality, $\{x,z\}\succeq_{sp} \{x,y,z\}$, the outcome obtained when $\sigma_i=(y,x,z,\ldots)$ is used.
Thus, we reach a contradiction, showing that Nom is not manipulable under $\info_{0}$.
\end{proof}

\begin{propi} \label{prop:nom-manip}
Nom is manipulable under $\info_{acc}$ (as well as under $\info_{\{acc\}}$).   Nom is also manipulable under $\info_{pl}$ (as well as under $\info_{\{pl\}}$). 
\end{propi}

\begin{proof}
First, for the case of acceptability information, consider a profile~$\prof{\pref}$ where alternative $a$ has $n$ acceptability points, alternative $b$ has $1$ acceptability point, and no other alternative has acceptability points. 
  Take a social planner with preference $\succ_{sp}$ that ranks $\{a,b\}$ first. 
  Then, the order $\prof{\sigma}$ such that $\sigma_i = (b,a,\ldots)$ for all $i\in N$ is an optimal strategy for the social planner. 
    Indeed, the first condition of manipulability is satisfied because $\text{Nom}(\prof{A_{\pref,\sigma}}) = \{a,b\} \succeq_{sp} X'$, for every $X'\subseteq \alt$ that can be the outcome of that profile under another order.
    For the second condition of manipulability, we need to find a profile with the same acceptability points as $\prof{\pref}$ for all alternatives, in which some order produces an outcome different from $\{a,b\}$. 
    For this purpose, consider a profile $\prof{\pref'}$ where $a$ is ranked first by every voter together with an order $\bar{\prof{\sigma}}$.
 Then $\text{Nom}(\prof{A_{\pref',\bar{\sigma}}}) = \{a\}$, and we are done.   

Second, for the case of plurality information,
the proof is similar to the proof of Theorem~\ref{thm:manip-weakandtotalunaniinfoplu}, without needing to use weak and total unanimity to justify the outcomes.
\end{proof}

Table~\ref{tab:anchor_strategy} provides a summary of all manipulability results we have obtained. Taken together, these results establish an informational hierarchy in the planner's ability to exploit anchoring. With full information, the planner can condition presentation orders on the voters' preference profile. With no information, manipulation is ruled out for broad classes of voting rules. Between these extremes, however, even limited aggregate information can restore manipulability. 

\begin{table}[H]
\centering
\begin{tabular}{@{}ll@{}}
\toprule
\textbf{Planner's Info} 
& \textbf{Manipulation Possible?} \\
\midrule
$\info_1$ 
& Yes, unless anchor-proof profile (Props~\ref{prop:troubledwaters}, \ref{prop:nomination-char}, \ref{prop:weakunacharacterisationprofile}) \\
\addlinespace
$\info_0$ 
& No (for W.unan.\ \& T.unan., Thm~\ref{thm:zero-info}; for Nom, Prop~\ref{prop:nom-info0manip}) \\
\addlinespace
$\info_{acc}$ 
& Yes (for W.unan., Anon., Neut., Thm \ref{thm:manip-unanimous}; for Nom, Prop~\ref{prop:nom-manip}) \\
\addlinespace
$\info_{pl}$ 
& Yes (for W.unan.\ \& T.unan., Thm~\ref{thm:manip-weakandtotalunaniinfoplu}; for \text{Nom}, Prop~\ref{prop:nom-manip}) \\

\bottomrule
\end{tabular}
\caption{Manipulability across informational assumptions.}
\label{tab:anchor_strategy}
\end{table}

To conclude this section, let us note that except for the case of zero information, all other information functions described here allow us to distinguish the status of the different alternatives in a profile. This is sensible in several contexts of collective decisions, such as polls of political elections and preliminary results of contests.

\section{Conclusion}\label{sec:conc}

Our analysis of anchor-proof voting rules has demonstrated how classical studies in social choice theory can be extended to accommodate systematic deviations from rational preference formation. We have considered voters with intrinsic preference rankings who are subject to an anchoring bias, influencing the formation of their approval ballots through the presentation order of alternatives. Our central result has shown that no non-constant approval-based voting rule is robust to the impact of such an anchoring bias. On the positive side, we have also found that a social planner who has no information about the voters' underlying preferences cannot safely use the presentation order to steer the voting outcome to her advantage.

Several directions for future work arise. One possibility is to consider alternative mechanisms for deriving anchor-biased ballots. For instance, trichotomous ballots could allow alternatives to be classified as approved, disapproved, or neutral, providing a more fine-grained representation of cognitive constraints. Randomization could also provide a way to mitigate some of the negative results, for example by assigning greater probability to orders that position objectively superior alternatives earlier. More generally, one could design voting rules that explicitly account for the presentation order, or study whether intrinsic preferences can be recovered from observed approval ballots and the orders under which they were generated.

The social planner's role also deserves further investigation. We have considered a particular class of planner preferences and a limited collection of information functions. Alternative objectives, such as targeting a specific outcome or treating alternatives symmetrically, may yield additional possibility results for anchor-proofness and manipulation.

Beyond anchoring, a systematic taxonomy of cognitive biases in voting could provide a basis for novel characterizations of voting rules, analogous to Fishburn's classification based on the information about the profile that a rule requires as input \citep{zwicker2016introduction}. Framing effects and social conformity biases are particularly promising candidates. More broadly, incorporating bounded rationality as a fundamental consideration in the design of collective decision procedures, rather than treating it as implementation noise, may lead to theoretically richer frameworks and more robust mechanisms for group decision-making. These considerations extend beyond voting to participatory budgeting, committee selection, peer review, recommendation systems, and resource allocation.


\section*{Statements and Declarations}
The authors have no competing interests to declare that are relevant to the content of this article.

\section*{Acknowledgments} We are grateful to the comments received by Ulle Endriss, by attendees of the 2026 meeting of the Society for the Social Choice and Welfare in Japan, of regular Economics seminars at the Maastricht University in the Netherlands and the Universidad Nacional del Sur in Argentina, as well as of the CRETE-2026 conference in Greece---their  feedback has been extremely valuable. We also acknowledge the CONDORCET project of PEPR Maths-Vives (ANR-24-EXMA-0001). 

\bibliographystyle{apalike}
\bibliography{anchor_bib}

\end{document}